\documentclass[runningheads]{llncs}
\usepackage{graphicx}
\usepackage{amsmath,amssymb}

\newcommand{\prev}{\mathit{prev}}
\newcommand{\forwd}{\mathit{fwd}}
\newcommand{\PSA}{\mathit{PSA}}
\newcommand{\PLCP}{\mathit{pLCP}}

\newtheorem{observation}{Observation}

\begin{document}
\title{Direct Linear Time Construction of Parameterized Suffix and LCP Arrays for Constant Alphabets}
\titlerunning{Linear Time Parameterized Suffix and LCP Arrays for Constant Alphabets}
\author{Noriki Fujisato \and
Yuto Nakashima \and
Shunsuke~Inenaga \and
Hideo~Bannai\orcidID{0000-0002-6856-5185} \and
Masayuki~Takeda
}
\authorrunning{Fujisato et al.}
\institute{Kyushu University, 744 Motooka, Nishi-ku, Fukuoka 819-0395, Japan\\
\email{\{noriki.fujisato,yuto.nakashima,inenaga,bannai,takeda\}@inf.kyushu-u.ac.jp}\\
}
\maketitle              \begin{abstract}
  We present the first worst-case linear time algorithm that directly computes the parameterized suffix and LCP arrays for constant sized alphabets. 
  Previous algorithms either required quadratic time or the parameterized suffix tree to be built first.
  More formally, 
  for a string over static alphabet $\Sigma$ and parameterized alphabet $\Pi$,
  our algorithm runs in $O(n\pi)$ time and $O(n)$ words of space,
  where $\pi$ is the number of distinct symbols of $\Pi$ in the string.
\end{abstract}
\section{Introduction}
Parameterized pattern matching is one of the well studied
``non-standard'' pattern matching problems
which was initiated by Baker~\cite{baker92-dup},
in an application to find duplicated code where variable names may be renamed.
In the parameterized matching problem, we consider strings over
an alphabet partitioned into two sets:
the parameterized alphabet $\Pi$ and the static alphabet $\Sigma$.
Two strings $x,y\in (\Pi\cup\Sigma)^*$ of length $n$ are said to
parameterized match (p-match), if one can be obtained
from the other with a bijective mapping over symbols of $\Pi$,
i.e.,
there exists a bijection $\phi: \Pi \rightarrow \Pi$
such that for all $1 \leq i \leq n$,
$x[i] = y[i]$ if $x[i] \in \Sigma$,
and $\phi(x[i]) = y[i]$ if $x[i] \in\Pi$.
For example,
if $\Pi = \{\mathtt{x},\mathtt{y},\mathtt{z}\}$ and $\Sigma = \{\mathtt{A},\mathtt{B},\mathtt{C}\}$,
strings $\mathtt{xxAzxByzBCzy}$ and $\mathtt{yyAxyBzxBCxz}$ p-match, since we can choose
$\phi(\mathtt{x}) = \mathtt{y}, \phi(\mathtt{y})= \mathtt{z}, $ and $\phi(\mathtt{z}) = \mathtt{x}$,
while strings $\mathtt{xyAzzByxBCz}$ and $\mathtt{yyAzxByxBCy}$
do not p-match, since
there is no such bijection on $\Pi$.
As parameterized matching captures the ``structure'' of the string,
it has also been extended to RNA structural matching~\cite{Shibuya2004}.

Baker introduced the so-called {\em prev encoding} of a p-string
which maps each symbol of the p-string that is in $\Pi$
to the distance to
its previous occurrence (or $0$ if it is the first occurrence), and
showed that two p-strings p-match if and only if their prev encodings are equivalent.
For example, the prev encodings for p-strings $\mathtt{xxAzxByzBCzy}$ and $\mathtt{yyAxyBzxBCxz}$ are both $(0,1,\mathtt{A},0,3,\mathtt{B},0,4,\mathtt{B},\mathtt{C},3,5)$.
Thus, the parameterized matching problem amounts to efficiently comparing the
prev encodings of the p-strings.

Using the prev encoding allows for the development of data structures
that mimic those of standard strings.
The central difficulty, in contrast with standard strings, is in coping with the following
property of prev encodings; a substring of a prev encoding is not necessarily
equivalent to the prev encoding of the corresponding substring.

Nevertheless, several data structures and algorithms have so far successfully been developed.
Baker proposed the parameterized suffix tree (PST), an analogue of the 
suffix tree for standard strings~\cite{DBLP:conf/focs/Weiner73},
and showed that for a string of length $n$,
it could be built in $O(n|\Pi|)$ time and  $O(n)$ words of space~\cite{DBLP:journals/jcss/Baker96}.
Using the PST for $T$, all occurrences of a substring in $T$ which
parameterized match a given pattern $P$ can be computed in $O(|P|(\log (|\Pi|+|\Sigma|)) + occ)$ time,
where $occ$ is the number of occurrences of the pattern in the text.
Kosaruju~\cite{DBLP:conf/focs/Kosaraju95} further improved the
running time of construction to $O(n\log(|\Pi|+|\Sigma|))$.
Furthermore, Shibuya~\cite{Shibuya2004} proposed an on-line algorithm for constructing the
PST that runs in the same time bounds.

Deguchi et al.~\cite{DBLP:conf/stringology/DeguchiHBIT08} proposed the parameterized suffix array (PSA).
Given the PST of a string, the PSA can be constructed in linear time, but as in the case for
standard strings, the direct construction of PSAs has been a topic of interest.j
Deguchi et al.~\cite{DBLP:conf/stringology/DeguchiHBIT08} showed a linear time algorithm 
for the special case of $|\Pi|=2$ and $\Sigma =\emptyset$.
I et al.~\cite{DBLP:conf/iwoca/TomohiroDBIT09} proposed a lightweight and practically efficient
algorithm for larger $\Pi$, but the worst-case time was still quadratic in $n$.
Beal and Adjeroh~\cite{DBLP:journals/jda/BealA12a} proposed an algorithm
based on arithmetic coding that runs in $O(n)$ time on average.
Furthermore, they claimed a worst-case running time of $o(n^2)$. However, the proved upperbound is
$O(n^2(\frac{\log(n-\log^{1+\varepsilon}n)}{\log^{1+\varepsilon}n}))$ for a very small $\varepsilon > 0$
(Corollary 27 of~\cite{DBLP:journals/jda/BealA12a}), 
so it is only slightly better than quadratic.

In this paper, we break the worst-case quadratic time barrier considerably, and
present the first worst-case linear time algorithm for constructing the parameterized
suffix and LCP arrays of a given p-string, when the number of distinct 
parameterized symbols in the string is constant.
Namely, our algorithm runs in $O(n\pi)$ time and $O(n)$ words of space,
where $\pi$ is the number of distinct symbols of $\Pi$ in the string.

Several other indices for parameterized pattern matching have been proposed.
Diptarama et al.~\cite{DBLP:conf/cpm/DiptaramaKONS17} and Fujisato et al.~\cite{DBLP:journals/corr/abs-1808-01071} proposed the parameterized position heaps (PPH), 
an analogue of the position heap for standard strings~\cite{EHRENFEUCHT2011100},
and showed that it could  be built in $O(n\log(|\Sigma|+|\Pi|))$ time and  $O(n)$ words of space.
Using the PPH for $T$, all occurrences of a substring in $T$ which
parameterized match a given pattern $P$ can be computed in $O(|P|(|\Pi|+\log (|\Pi|+|\Sigma|)) + occ)$ time,
where $occ$ is the number of occurrences of the pattern in the text.
Parameterized BWT's have been proposed in \cite{DBLP:conf/soda/0002ST17}.
Also, paramterized text index with one wildcard was proposed in~\cite{8712707}.

\section{Preliminaries}
For any set $A$ of symbols, $A^*$ denotes the set of
strings over the alphabet $A$.
Let $|x|$ denote the length of a string $x$. The empty string is denoted by $\varepsilon$.
For any string $w\in A^*$, if
$w = xyz$ for some (possibly empty) $x,y,z\in A^*$, $x,y,z$ are respectively called
a {\em prefix}, {\em substring}, {\em suffix} of $w$.
When $x,y,z \neq w$, they are respectively called a
{\em proper} prefix, substring, and suffix of $w$.
For any integer $1 \leq i \leq |x|$, $x[i]$ denotes the $i$th symbol in $x$,
and for any $1\leq i \leq j \leq |x|$, $x[i..j]= x[i]\cdots x[j]$.
Let $\prec$ denote a total order on $A$, as well as
the lexicographic order it induces.
For two strings $x,y\in A^*$,
$x\prec y$
if and only if $x$ is a proper prefix of $y$, or there is some
position $1 \leq k \leq \min \{ |x|,|y| \}$ such that
$x[1..k-1] = y[1..k-1]$ and $x[k] \prec y[k]$.

Let $\Pi$ and $\Sigma$ denote disjoint sets of symbols.
$\Pi$ is called the parameterized alphabet, and $\Sigma$ is called the static alphabet.
A string in $(\Pi\cup\Sigma)^*$ is sometimes called a p-string.
Two p-strings $x,y \in (\Pi\cup\Sigma)^*$ of equal length are said to
{\em parameterized match},
denoted $x\approx y$,
if
there exists a bijection $\phi:\Pi\rightarrow\Pi$, such that for all
$1\leq i \leq |x|$,
$x[i] = y[i]$ if $x[i]\in\Sigma$, and $\phi(x[i]) = y[i]$ if $x[i]\in\Pi$.

The {\em prev encoding} of a p-string $x$ of length $n$ is the string $\prev(x)$ over
the alphabet $\Sigma\cup\{ 0, \ldots, n-1\}$ defined as follows:
\[
  \prev(x)[i] =
  \begin{cases}
    x[i]  & \text{if $x[i] \in \Sigma$,}\\
    0     & \text{if $x[i] \in \Pi$ and $x[i]\neq x[j]$ for any $1 \leq j < i$,}\\
    i-j   & \text{if $x[i] \in \Pi$, $x[i]=x[j]$ and $x[i]\neq x[k]$ for any $j < k < i$.}
  \end{cases}
 \]
For example, if $\Pi=\{\mathtt{s},\mathtt{t},\mathtt{u}\}$,
$\Sigma=\{\mathtt{A}\}$ and p-string
$x=\mathtt{ssuAAstuAst}$, then
$\prev(x)=(0,1,0,\mathtt{A},\mathtt{A},4,0,5,\mathtt{A},4,4)$.
Baker showed that $x \approx y$ if and only if $\prev(x) = \prev(y)$~\cite{DBLP:journals/siamcomp/Baker97}.
We assume that $\Pi$ and $\Sigma$ are disjoint integer alphabets, where
$\Pi = \{ 0, \ldots, n^{c_1}\}$ for some constant $c_1\geq 1$
and $\Sigma = \{ n^{c_1} + 1, \ldots, n^{c_2}\}$ for some constant $c_2\geq 1$.
This way, we can distinguish whether a symbol of a given prev encoding belongs to
$\Sigma$ or not.
Also, given p-string $x$ of length $n$, we can compute $prev(x)$ in $O(n)$ time and space,
by sorting the pairs $(x[i],i)$ using radix sort, followed by a simple scan of the result.

The following are the data structures that we consider in this paper.
\begin{definition}[Parameterized Suffix Array~\cite{DBLP:conf/stringology/DeguchiHBIT08}]
  The parameterized suffix array of a p-string $x$ of length $n$,
  is an array $\PSA[1..n]$ of integers such that
  $\PSA[i] = j$ if and only if
  $\prev(x[j..n])$ is the $i$th lexicographically smallest string in
  $\{ \prev(x[i..n])\mid i = 1,\ldots, n \}$.
\end{definition}

\begin{definition}[Parameterized LCP Array~\cite{DBLP:conf/stringology/DeguchiHBIT08}]
  The parameterized LCP array of a p-string $x$ of length $n$,
  is an array $\PLCP[1..n]$ of integers such that
  $\PLCP[1] = 0$, and
  $\PLCP[i]$, for any $i \in \{ 2, \ldots, n\}$, is the longest common prefix between
  $\prev(x[\PSA[i-1]..n])$ and $\prev(x[\PSA[i]..n])$.
\end{definition}

The difficulty when dealing with the prev encoding of suffixes of a string,
is that they are not necessarily the suffixes of the prev encoding of the string.
It is important to notice however, that, given the prev encoding $\prev(x)$ 
of the whole string $x$, any value specific of the prev encoding of an arbitrary suffix of $x$
can be retrieved in constant time, i.e., for any $1 \leq i \leq n$ and $1 \leq k \leq n-i+1$,
\[
  \prev(x[i..n])[k] = \begin{cases}
    0 & \text{if $x[k']\in\Pi$ and $\prev(x)[k'] > k$.}\\
    \prev(x)[k'] & \text{otherwise}
  \end{cases}
  \]
  where $k' = i+k-1$.
The critical problem for suffix sorting is that
even if two prev encodings $\prev(x[i..n])$ and $\prev(x[j..n])$
share a common prefix and satisfies $\prev(x[i..n]) \prec \prev(x[j..n])$,
it may still be that $\prev(x[j+1..n]) \prec \prev(x[i+1..n])$.

Fig.~\ref{figure:fig1} shows an example of $\PSA$ and $\PLCP$ for the string $\mathtt{stssAtssAs}$.
For example,
we have that $\prev(x[6..10]) \prec \prev(x[1..10])$,
which share a common prefix of length $2$,
yet $\prev(x[2..10]) \prec \prev(x[7..10])$.

 \begin{figure}[htbp]
\begin{center}
\begin{tabular}{|c|c|l|c|}\hline
  $i$ & $\PSA[i]$ & $\prev(x[\PSA[i]..|r|])$ & $\PLCP[i]$\\\hline
  1 & 10 & 0  & 0\\\hline
  2 & 6 & 0 0 1 $\mathtt{A}$ 2 & 1\\\hline
  3 & 2 & 0 0 1 $\mathtt{A}$ 4 3 1 $\mathtt{A}$ 2 & 4\\\hline
  4 & 1 & 0 0 2 1 $\mathtt{A}$ 4 3 1 $\mathtt{A}$ 2 & 2\\\hline
  5 & 3 & 0 1 $\mathtt{A}$ 0 3 1 $\mathtt{A}$ 2 & 1\\\hline
  6 & 7 & 0 1 $\mathtt{A}$ 2 & 3\\\hline
  7 & 4 & 0 $\mathtt{A}$ 0 3 1 $\mathtt{A}$ 2 & 1\\\hline
  8 & 8 & 0 $\mathtt{A}$ 2 & 2\\\hline
  9 & 9 & $\mathtt{A}$ 0 & 0\\\hline
  10 & 5 & $\mathtt{A}$ 0 0 1 $\mathtt{A}$ 2 & 2\\\hline
\end{tabular}
\end{center}
  \caption{An example of the parameterized suffix and LCP arrays for a p-string $x = \mathtt{stssAtssAs}$, where $\Sigma=\{\mathtt{A}\},\Pi=\{\mathtt{s},\mathtt{t}\}$.}\label{figure:fig1}
  \end{figure}
 \section{Algorithms}
In this section we describe our algorithms for constructing the parameterized suffix and LCP arrays.
First, we mention a simple observation below.

From the definition of $\prev(x)$, we have that $\prev(x)[i] = 0$ for some position $i$
if and only if $i$ is the first occurrence of symbol $x[i]\in \Pi$.
Therefore, the following observation can be made.
\begin{observation} For any p-string $x$, the prev encoding $\prev(x')$ of any substring $x'$ of $x$ contains at most $\pi $ positions that are $0$'s, where $\pi$ is the number of distinct symbols of $\Pi$ in $x$.
\end{observation}
\subsection{$\PSA$ Construction}
Based on this observation, we can see that
the prev encoding of each suffix $x[i..n]$ can be
partitioned into $z_i+1 \leq \pi +1$ {\em blocks},
where $z_i$ is the number of $0$'s in $\prev(x[i..n])$,
and the $j$th block is the substring of $\prev(x[i..n])$ that ends at the
$j$th $0$ in $\prev(x[i..n])$ for $j = 1,\ldots,z_i$,
and the (possibly empty) remaining suffix for $j = z_i+1$.
For technical reasons, we will append $0$ to the last block as well.
That is, we can write
\begin{equation}\label{equation:blocks}
  \prev(x[i..n])0 = B_{i,1}\cdots B_{i,z_i+1}
\end{equation}
where, $B_{i,j}$ denotes the $j$th block of $\prev(x[i..n])$.
Furthermore, for each $j$, let $B_j$ denote the set of all $j$th blocks for all $i = 1,\ldots, n$,
and let $C_{i,j}$ denote the lexicographic rank of $B_{i,j}$ in $B_j$.
Finally, let $C_i$ denote the
string over the alphabet $\{ 1,\ldots, n\}$ obtained
by renaming each block $B_{i,j}$ of the string $\prev(x[i..n])0$ with its lexicographic rank $C_{i,j}$.
More formally,
\begin{eqnarray*}
  B_j &=& \{ B_{i,j} \mid i = 1,\ldots, n\}\\
  C_{i,j} &=& |\{ B_{i',j}\mid B_{i',j}\prec B_{i,j}\}|+1\\ C_i &=& C_{i,1}\cdots C_{i,z_{i}+1}.
\end{eqnarray*}

\begin{lemma}\label{lemma:reduction}
  For any $1 \leq i_1, i_2 \leq n$,
 \[
   \prev(x[i_1..n])\prec\prev(x[i_2..n]) \iff C_{i_1} \prec C_{i_2}
 \]
\end{lemma}
\begin{proof}
Notice that $0$ is the smallest symbol in the two strings,
so
\begin{eqnarray*}
  \prev(x[i_1..n]) \prec \prev(x[i_2..n]) &\Leftrightarrow&
\prev(x[i_1..n])0 \prec \prev(x[i_2..n])0\\
&\Leftrightarrow& B_{i_1,1}\cdots B_{i_1,z_{i_1}+1} \prec
B_{i_2,1}\cdots B_{i_2,z_{i_2}+1}.
\end{eqnarray*}
Also notice that since any block must end with a $0$,
if two blocks are not identical, it holds that one
cannot be a prefix of the other.
Thus, if
$B_{i_1,1}\cdots B_{i_1,z_{i_1}+1} \prec
B_{i_2,1}\cdots B_{i_2,z_{i_2}+1}$,
this implies that there is some block $k$ such that
$B_{i_1,j} = B_{i_2,j}$, for all $1 \leq j < k$,
and $B_{i_1,k} \prec B_{i_2,k}$, where
$B_{i_1,k}$ is not a prefix of $B_{i_2,k}$.
By definition,
$B_{i_1,k} \preceq B_{i_2,k} \Leftrightarrow C_{i_1,k} \leq C_{i_2,k}$.
Therefore, we have,
$B_{i_1,1}\cdots B_{i_1,z_{i_1}+1}\prec B_{i_2,1}\cdots B_{i_2,z_{i_2}+1} \Leftrightarrow C_{i_1} \prec C_{i_2}$.
\qed
\end{proof}

From Lemma~\ref{lemma:reduction}, the problem of lexicographically sorting
the set of strings $\{ \prev(x[1..n]), \ldots, \prev(x[n..n])\}$
reduces to the problem of lexicographically sorting
the set of strings $\{ C_1, \ldots, C_n\}$.
The latter can be done in $O(n\pi )$ time using radix sort,
since the strings are over the alphabet $\{1,\ldots,n\}$ and the total length of the
strings is at most $n\pi $.

What remains is to to compute $C_{i,j}$ for all $i,j$
in the same time bound. A problem is that the total length of all
$B_{i,j}$ is $\Theta(n^2)$, so we cannot afford to naively process
all of them.

Denote by $b_{i,j}$ and $e_{i,j}$ the beginning and end positions
of $B_{i,j}$ with respect to their (global) position in $x$.
Note that for any $1 \leq i \leq n$, we have $b_{i,1} = i$,
and $b_{i,j} = e_{i,j-1}+1$ for all $2\leq j\leq z_{i}+1$.
Our algorithm depends on the following simple yet crucial lemma.
\begin{lemma}\label{lemma:block_composition}
  For any $1< i \leq n$ and $1 \leq j \leq z_{i}+1$,
  we have that either
  \begin{enumerate}
  \item $b_{i,j} = e_{i-1,j}+1$, or,
  \item $b_{i,j} \geq b_{i-1,j}$, $e_{i,j} = e_{i-1,j}$, and $B_{i,j}$ is a suffix of $B_{i-1,j}$
 \end{enumerate}
  holds.
\end{lemma}
\begin{proof}
  If $x[i-1] \in \Sigma$, then,
  $\prev(x[i..n])$ is a suffix of $\prev(x[i-1..n])$,
  i.e.,
  $\prev(x[i..n]) = \prev(x[i-1..n])[2..|n-i+2|]$
  and $\prev(x[i-1..n])[1] \neq 0$.
  Thus,
  $B_{i,1}$ is a suffix of $B_{i-1,1}$, and
  $B_{i,j} = B_{i-1,j}$ for all $2 \leq j \leq z_i$
  and the second case of the claim holds.

  If $x[i-1] \in \Pi$,
  the values in $\prev(x[i..n])$ are equivalent to the corresponding values
  of $\prev(x[i-1..n])[2..|n-i+2|]$, except possibly at some (global) position $k\geq i$
  when there is a second occurrence of the symbol $x[i-1]$ at $x[k]$ which becomes
  the first occurrence in $x[i..n]$. (In other words, the value corresponding to $x[k]$
  in $\prev(x[i-1..n])$ is $k-i+1$.)
  Since there is no previous occurrence of $x[i-1]$ in $x[i-1..n]$, $\prev(x[i-1..n])[1]=0$.
  The situation is depicted in Fig.~\ref{lemma:block_composition}.

  \begin{figure}[htbp]
  \begin{center}
    \includegraphics[width=\textwidth]{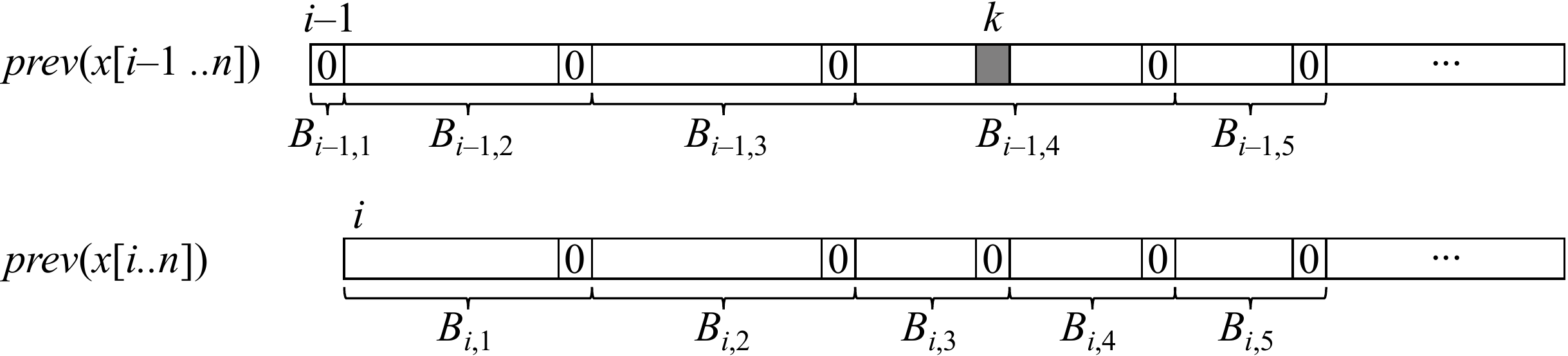}
  \end{center}
  \caption{A case in the proof of Lemma~\ref{lemma:block_composition},
  where $x[i]\in\Pi$, and $x[k]$ is the first occurrence of $x[i]$ in $x[i+1..n]$.
  The value corresponding to (global) position $k$ in $\prev(x[i..n])$, shown as a shaded box,
  is $k-i$, while it is $0$ in $\prev(x[i+1..n])$.
  All other values in $\prev(x[i..n])$ and $\prev(x[i+1..n])$
  at the same (global) position are equivalent. }
  \label{fig:block_composition}
  \end{figure}

  Let $B_{i-1,j'}$ be the block of $\prev(x[i-1..n])$ that contains (global) position $k$.
  Because, as mentioned previously,
  $\prev(x[i]..n)$ and $\prev(x[i-1..n])[2..|n-i+2|]$ are equivalent
  except for the value corresponding to (global) position $k$,
  the block structure of $\prev(x[i-1..n])$ is preserved in $\prev(x[i..n])$,
  except that (1) the first block $B_{i-1,1}$ disappears, and
  (2) the block $B_{i-1,j'}$ is split into two blocks, corresponding to $B_{i,j'-1}$ and $B_{i,j'}$.
  Therefore, the first case of the claim is satisfied for $1 \leq j \leq j'$,
  since $b_{i,j} = b_{i-1,j+1} = e_{i-1,j}+1$ for any $1 \leq j <  j'$.
  Also, we can see that the second case of the claim is satisfied
  for $j' \leq j \leq z_{i}$,
  since
  $B_{i,j'}$ is a suffix of $B_{i-1,j'}$,
  and $B_{i,j} = B_{i-1,j}$ for $j' < j \leq z_{i}$.

  Finally, the case when such $k$ does not exist can be considered to be included above
  by simply assuming we are looking at a prefix of a longer string and $k > |x|, j' > z_{i}$,
  since the prev encoding is preserved for prefixes, i.e.,
  the prev encoding of a prefix of any p-string $y$ is equivalent to the corresponding
  prefix of the prev encoding of $y$.
  Thus, the lemma holds.
  \qed
\end{proof}

Lemma~\ref{lemma:block_composition} implies that if we fix some $j$,
we can represent $B_{i,j}$ for all $i$, as suffixes (in the standard sense)
of strings of total length $O(n)$.

\begin{corollary}\label{corollary:blocksuperstring}
  For any $j$,
  there exists a set of strings $S_j$ with total length $n+1$
  over the alphabet $\Sigma\cup\{0,\ldots,n-1\}$
  such that $B_{i,j}$ is a suffix of some string in $S_j$ for all $i \in \{1, \ldots, n\}$.
\end{corollary}
\begin{proof}
  We include $B_{i,j}$ in $S_j$, if $i = 1$, or,
  if $i > 1$ and $B_{i,j}$ satisfies the first case of Lemma~\ref{lemma:block_composition}.
  Since the first case implies that the (global) positions $[b_{i-1,j}..e_{i-1,j}]$ and $[b_{i,j}...e_{i,j}]$
  are disjoint, the total length of strings in $S_j$ is at most $n+1$ (including the $0$ appended to $B_{i,z_{i}+1}$).
  On the other hand, if $B_{i,j}$ satisfies the second case is,
  it is a suffix of an already included string.\qed
\end{proof}

Thus, computing $C_{i,j}$ for all $i$ can be done by
computing the generalized suffix array for the set $S_j$.
This can be done in $O(n)$ time given $S_j$~\cite{DBLP:journals/jda/KimSPP05,DBLP:journals/jda/KoA05,DBLP:journals/jacm/KarkkainenSB06} and 
thus, for all $j$, the total is $O(n\pi )$ time.
\begin{theorem}\label{theorem:computingPSA}
  The parameterized suffix array of a p-string of length $n$
  can be computed in $O(n\pi )$ time and $O(n)$ space.
\end{theorem}
\begin{proof}
  We compute a {\em forward encoding} of $x$, analogous to the prev encoding, defined as follows
  \[
    \forwd(x)[i] =
    \begin{cases}
      x[i]  & \text{if $x[i] \in \Sigma$,}\\
      \infty     & \text{if $x[i] \in \Pi$ and $x[i]\neq x[j]$ for any $i< j \leq n$,}\\
      j-i   & \text{if $x[i] \in \Pi$, $x[i]=x[j]$ and $x[i]\neq x[k]$ for any $i < k < j$.}
    \end{cases}
    \]
  This is done once, and can be computed in $O(n)$ time.
  Next, for any fixed $j$, we show how to compute the set $S_j$ in linear time.
  This is done by using $\forwd$ and Lemma~\ref{lemma:block_composition}.
  We can first scan $\prev(x)$ to obtain $B_{1,j}$.
  Suppose for some $i \geq 2$, we know the beginning and end positions $b_{i-1,j}$, $e_{i-1,j}$ of $B_{i-1,j}$.
  Notice that when $x[i]\in\Pi$, $k$ in the proof of Lemma~\ref{lemma:block_composition} 
  is $i+\forwd(x)[i-1]-2$. 
  Based on this value, we know that if $k < b_{i-1,j}$, then $B_{i,j} = B_{i-1,j}$ and if
  $b_{i-1,j} \leq k \leq e_{i-1,j}$ $B_{i,j}$ is a suffix of $B_{i-1,j}$, 
  which corresponds to the second case of Lemma~\ref{lemma:block_composition}.
  When $k > e_{i-1,j}$, this corresponds to the first case of Lemma~\ref{lemma:block_composition},
  so we scan $\prev(x[i..n])$ starting from position corresponding to the global position
  $b_{i,j} = e_{i-1,j}+1$ (i.e., $e_{i-1,j}-i$ in $\prev(x[i..n])$) until we find the first $0$,
  which gives us $B_{i,j}$ which we include in $S_j$.
  Since we only scan each position once, the total time for computing $S_j$ is $O(n)$.

  The time complexity follows from arguments for sorting $C_j$ based on radix sort.
  Since, for a single step of the radix sort,
  we only require the values $C_{i,j}$ for a fixed $j$ and
  all $1\leq i\leq n$ and from Corollary~\ref{corollary:blocksuperstring},
  the space complexity is $O(n)$. \qed
\end{proof}

\subsection{$\PLCP$ Construction}
Given $\PSA$, we can construct $\PLCP$ as follows in $O(n\pi )$ time and $O(n)$ space.
We recompute $S_j$ for $j = 1, \ldots, \pi $,
and each time process it for LCE queries, so that the longest common prefix between
$B_{i_1,j}$ and $B_{i_2,j}$ for some $1 \leq i_1,i_2 \leq n$ can be computed in constant time.
This can be done in time linear in the total length of $S_j$, so in $O(n\pi )$ total time for all $j$.
We compute the longest common prefix between each adjacent suffix in $\PSA$ block by block.
Since each block takes constant time, and
there are $O(\pi)$ blocks for each suffix,
the total is $O(n\pi )$ time for all entries of the $\PLCP$ array.
The space complexity is $O(n)$ since, as for the case of $\PSA$ construction,
we only process the $j$th block at each step.

\bibliographystyle{splncs04}
\bibliography{refs}

\end{document}